\documentclass[runningheads,envcountsame]{llncs}

\usepackage{graphicx}
\usepackage{xcolor}
\usepackage{rotating}

\usepackage[utf8]{inputenc}
\usepackage{amsmath}
\usepackage{amssymb}
\let\doendproof\endproof
\renewcommand\endproof{~\hfill\qed\doendproof}

\usepackage{subfig}
\usepackage{enumerate}
\usepackage{algorithm}
\usepackage[noend]{algorithmic}
\usepackage{complexity}
\newclass{\threesat}{3SAT}

\usepackage[urlcolor=blue]{hyperref}
\hypersetup{colorlinks=true,linkcolor=blue,citecolor=blue,linktocpage}
\title{Extending Simple Drawings\thanks{This work was started at the Crossing Numbers Workshop 2016 in Strobl (Austria). M.D. was partially supported by NSERC. I.P. is supported by the Austrian Science Fund (FWF): W1230. This project has received funding from the European Union’s Horizon 2020 research and innovation
programme under the Marie Skłodowska-Curie grant agreement No 754411.} 
}

\author{Alan Arroyo\inst{1} \and
	Martin Derka\inst{2} \and
	Irene Parada\inst{3}}
\authorrunning{A. Arroyo et al.}
\institute{
	IST Austria, Klosterneuburg, Austria\\
	\email{alanmarcelo.arroyoguevara@ist.ac.at}
	 \and
	University of Waterloo, Ontario, Canada\\
	\email{mderka@uwaterloo.ca}\\
	 \and
	Graz University of Technology, Graz, Austria\\
	\email{iparada@ist.tugraz.at}
}

\spnewtheorem{observation}[theorem]{Observation}{\bfseries}{\itshape}

\begin{document}

\maketitle

\begin{abstract}
	Simple drawings of graphs are those in which 
	each pair of edges share at most one point, 
	either a common endpoint or a proper crossing.  
	In this paper we study the problem of extending a simple drawing $D(G)$ of a graph 
	$G$ by inserting a set of edges from the complement of $G$ into $D(G)$ 
	such that the result is a simple drawing. 
	In the context of rectilinear drawings, the problem is trivial.  
	For pseudolinear drawings, the existence of such an extension follows from Levi's enlargement lemma.
	In contrast, 
	we prove that deciding if a given set of edges can be inserted into a simple drawing is \NP-complete.
	Moreover, we show that the maximization version of the problem is \APX-hard. 
	We also present a polynomial-time algorithm 
	for deciding whether one edge $uv$ can be inserted into $D(G)$ 
	when $\{u,v\}$ is a dominating set for the graph $G$.
	
\end{abstract}

\section{Introduction}
A \emph{simple drawing} of a graph $G$ (also known as \emph{good drawing} or as \emph{simple topological graph} in the literature) 
 is a drawing $D(G)$ of $G$ in the plane such that
every pair of edges share at most one point that 
is either a proper crossing (no tangent edges allowed) or an endpoint.
Moreover, no three edges intersect in the same point
and edges must neither self-intersect nor contain other vertices than their endpoints. 
 Simple drawings, despite often considered in the study of crossing numbers, 
 have basic aspects that are yet unknown.

The long-standing conjectures on the crossing numbers of $K_n$ and $K_{n,m}$,  known as the Harary-Hill and Zarankiewicz’s conjectures, respectively, have drawn particular interest in the study of simple  drawings of complete and complete bipartite graphs.
The intensive study of these conjectures has produced deep results about simple drawings of $K_n$ \cite{Kyncl15,Pach2003} and $K_{n,m}$~\cite{CardinalFelsner16}. 

In contrast to our knowledge about $K_n$, little is known about simple drawings of general graphs. 
In  \cite{saturated} it was observed that, when studying simple drawings of general  graphs, 
it is  natural to try to extend them, by inserting the missing edges between non-adjacent vertices. 
One of the main results in this paper suggests that there is no hope for efficiently deciding when such 
operation can be performed. 

The complement $\overline{G}$ of a graph $G$ is the graph 
with the same vertex set as $G$ 
and where two distinct vertices are adjacent 
if and only they are not adjacent in $G$. 
Given a simple drawing $D(G)$ of a graph $G=(V,E)$ and a subset $M$ of {\em candidate edges} 
from $\overline{G}$, 
an \emph{extension} of $D(G)$ with $M$ 
is a simple drawing $D'(H)$ of the graph $H = (V,E\cup M)$ 
that contains $D(G)$ as a subdrawing. 
If such an extension exists, 
then we say that $M$ can be \emph{inserted} into $D(G)$. 

Given a simple drawing,  
an extension with one given edge is not always possible, as 
shown by Kyn\v{c}l~\cite{Kyncl13} 
(in Fig.~\ref{fig:kyncl} the edge $uv$ cannot be inserted, because $uv$ would cross an edge incident either to $u$ or to $v$). 
We can extend this example to a simple drawing of $K_{2,4}$ (Fig.~\ref{fig:noK24}) 
and we can then use it to construct drawings of $K_{n,m}$ with larger values of $m$ and $n$ 
in which an edge $uv$ cannot be inserted. 
Moreover, Kyn\v{c}l's drawing can be extended to a simple drawing of $K_6$ minus one edge where the missing edge cannot be inserted (Fig.~\ref{fig:noK6}).  From this drawing one can construct drawings of $K_{n}$ with $n\geq 6$ minus one edge where the only 
missing edge cannot be inserted.

Extensions, by inserting both vertices and edges,  
have received a great deal of attention in the last decade, specially  
 for (different classes of) plane drawings~\cite{ext_plane,Planarstraightline_Bagheri_2010,PartialConstrainedLevel_Brueckner_2017,ext_up_2019,Kuratowskitypetheorem_Jelinek_2013,ExtendingConvexPartial_Mchedlidze_2015,ext_straight_06}. 
It has also been of interest to study crossing number questions on planar graphs with one additional edge~\cite{AddingOneEdge_Cabello_2013,InsertingEdgePlanar_Gutwenger_2005,crossingnumbercubic_Riskin_1996}. 
Note that the term \emph{augmentation} has 
also been used in the literature for the similar problem of 
inserting edges and/or vertices to a graph~\cite{augment}.
Extensions of simple drawings have been previously considered in the context of \emph{saturated} drawings, 
that is, drawings where no edge can be inserted~\cite{hajnal2015saturated,saturated}.
\begin{figure}[tb]
	\centering
	\subfloat[\label{fig:kyncl}Example by Kyn\v{c}l~\cite{Kyncl13}.]{\includegraphics[page =1]{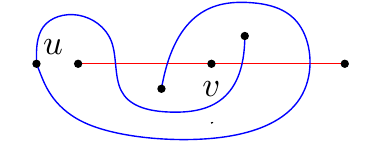}}	
	\quad
	\subfloat[\label{fig:noK24}Drawing of $K_{2,4}$.]{\includegraphics[page = 2]{kyncl_nonexten2.pdf}}
	\quad
	\subfloat[\label{fig:noK6}Drawing of $K_{6}-uv$.]{\includegraphics[page = 3]{kyncl_nonexten2.pdf}}
	\caption{Drawings in which the edge $uv$ cannot be inserted.}
	\label{fig:6vert}
\end{figure}

\paragraph{Our Contribution} 
We study the computational complexity of extending a simple drawing $D(G)$ of a graph $G$.
In Section~\ref{sec:Addk}, 
we show that deciding if $D(G)$ 
can be extended with a set $M$ of candidate edges 
is \NP-complete. 
Moreover, in Section~\ref{sec:APXhard}, 
we prove that finding the largest subset of edges from $M$ 
that extend $D(G)$ is \APX-hard.
Finally, in Section \ref{sec:adding_one_edge}, 
we present a polynomial-time algorithm to decide 
whether an edge $uv$ can be inserted into $D(G)$
when $\{u,v\}$ is a dominating set for $G$.

\section{Inserting a given set of edges is \NP-complete}
\label{sec:Addk}

In this section we prove the following result: 

\begin{theorem}
	Given a simple drawing $D(G)$ of a graph $G=(V,E)$ 
	and a set $M$ of edges of the complement of $G$, 
	it is \NP-complete to decide if $D(G)$ can be extended with the set $M$.
	\label{thm:addk}
\end{theorem}

Notice first that the problem is in \NP, since it can be described combinatorially. 
Our proof of Theorem~\ref{thm:addk} is based on a reduction from \emph{monotone \threesat}~\cite{pla-mon-rect-3SAT}.
An instance of that problem consists of a Boolean formula $\phi$ in 3-CNF 
with a set of variables $X=\{x_1,\ldots, x_n\}$ and a set of clauses $K=\{C_1,\ldots, C_m\}$. 
Moreover, in each clause either all the literals are positive ({\em positive clause}) or they are all negative ({\em negative clause}). 
The {\em bipartite graph $G(\phi)$ associated to $\phi$} is 
the graph with vertex set $X\cup K$ 
and where a variable $x_i$ is adjacent to a clause $C_j$ if and only if $x_i\in C_j$ or  $\overline{x_i}\in C_j$.

\begin{figure}[tb]
	\subfloat[Variable gadget $\mathcal{X}$.\label{fig:var}]{
		\includegraphics[page=1]{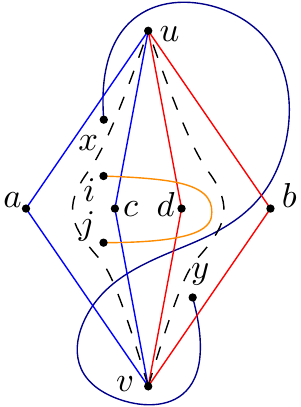}
	}
	\quad
	\subfloat[Clause gadget $\mathcal{C}$.\label{fig:clause}]{
		\includegraphics[page=2]{reduction_exactlyk}
	}
	\quad
	\centering 
	\subfloat[Wire gadget $\mathcal{W}$.\label{fig:wire}]{
		\includegraphics[page=3]{reduction_exactlyk}
	}
	\caption{Basic gadgets for the proof of Theorem~\ref{thm:addk}.}\label{fig:gadgets3SAT}	
\end{figure}

We now show how to construct a simple drawing from a given formula. 
We start by introducing our three basic gadgets, 
the \emph{variable gadget}, the \emph{clause gadget}, and the \emph{wire gadget}, 
shown in Fig.~\ref{fig:gadgets3SAT}. 

The variable gadget contains two nested cycles, 
$avbu$ on the outside and $cvdu$ on the inside,   
drawn in the plane without any crossings.
Two additional vertices $x$ and $y$ are drawn in the interior of $avcu$ and $dvbu$, respectively. 
They are connected with an edge that, starting in $x$, crosses the edges $au$, $ub$, $dv$, $cv$, $av$, and $vb$, in this order, and ends in $y$.  
Another two vertices $i$ and $j$ are drawn inside the region in the interior of $avcu$ that is incident to $x$. 
They are connected with an edge that, starting in $i$, crosses the edges $uc$, $ud$, $vd$, and $vc$, in this order, and ends in $j$; 
see Fig.~\ref{fig:var}. 
Notice that the edge $uv$ can be inserted only in two possible regions: 
either inside the cycle $avcu$ or inside the cycle $dvbu$. 
Drawing the edge $uv$ in any other region would force it to cross $uj$ or $xy$ more than once. 
The clause gadget and the wire gadget are similarly defined; 
see Fig.~\ref{fig:gadgets3SAT}b--c. 

In each of these three gadgets shown in Fig.~\ref{fig:gadgets3SAT}, 
the edge $uv$ can only be inserted in the regions where the dashed arcs are drawn. 
In the rest of the paper, when we refer to the {\em regions} in a gadget we mean these regions where the edge $uv$ can be inserted.

In a variable gadget, 
these regions encode the truth assignment of the corresponding variable $x_i$:  
inserting the edge $uv$ in the left region 
corresponds to the assignment $x_i=\texttt{true}$, 
while inserting it in the right region corresponds to $x_i=\texttt{false}$. 
We call these left and right regions in a variable gadget 
the \texttt{true} and \texttt{false} regions, respectively. 
In a clause gadget, 
each of the three regions is associated to a literal in the corresponding clause.  
Wire gadgets propagate the truth assignment of the variables to the clauses. 
They are drawn between the gadgets corresponding to clauses and variables that are incident in $G(\phi)$. 
The idea is that if an assignment makes a literal not satisfy a clause, 
then the edge $uv$ in the wire gadget 
blocks the region in the clause gadget corresponding to that literal 
by forcing $uv$ to cross that region twice.

Let $w^{(\mathcal{G})}$ denote vertex $w$ in gadget $\mathcal{G}$. 
The following lemma shows that we can get the desired behavior with a wire gadget connecting a variable gadget and a clause gadget. 
The precise placement of a wire gadget with respect to the variable gadget and the clause gadget that it connects is 
illustrated in Fig.~\ref{fig:red3SAT}. 

\begin{figure}[tb]
	\centering 
	\includegraphics[page=1]{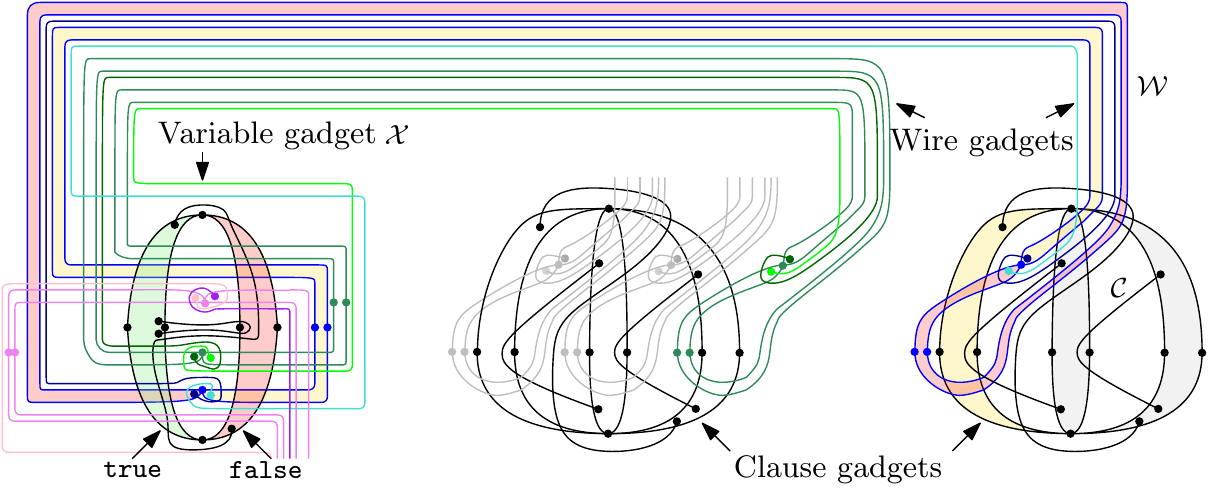}
	\caption{Reduction from monotone \threesat.}\label{fig:red3SAT}	
\end{figure}

\begin{lemma}\label{lem:wire}
	We can combine a variable gadget $\mathcal{X}$, a clause gadget $\mathcal{C}$, and a wire gadget $\mathcal{W}$
	to produce a simple drawing  with the following properties. 
	\begin{itemize}
		\item If $u^{(\mathcal{X})}v^{(\mathcal{X})}$ is inserted in the \texttt{false} region in $\mathcal{X}$, then inserting $u^{(\mathcal{W})}v^{(\mathcal{W})}$ 
		prevents $u^{(\mathcal{C})}v^{(\mathcal{C})}$ from being inserted in one specified target region in $\mathcal{C}$. 
		\item If $u^{(\mathcal{X})}v^{(\mathcal{X})}$ is inserted in the \texttt{true} region in $\mathcal{X}$, then we can insert $u^{(\mathcal{W})}v^{(\mathcal{W})}$ in a way such that $u^{(\mathcal{C})}v^{(\mathcal{C})}$ can then be inserted in any region in $\mathcal{C}$.
	\end{itemize}	
\end{lemma}

\begin{proof}
We start with a drawing of the variable gadget $\mathcal{X}$ and the clause gadget $\mathcal{C}$ such that the two gadgets are drawn 
 on a line and they are disjoint.   
A representation of how the wire gadget is then inserted is shown in Fig.~\ref{fig:red3SAT}.  
In this proof we focus on the wire gadget drawn with blue edges and vertices. 

In Fig.~\ref{fig:red3SAT}, 
gadget $\mathcal{X}$ lies to the left of gadget $\mathcal{C}$. 
The \texttt{true} and \texttt{false} regions in $\mathcal{X}$ 
are shaded in green and red, respectively. 
We assume that the target region in $\mathcal{C}$ is the leftmost one, 
shaded in yellow. 
The left and right regions in the wire gadget are shaded in red and yellow, respectively. 

If the edge $u^{(\mathcal{X})}v^{(\mathcal{X})}$ is inserted in the \texttt{false} region in $\mathcal{X}$ 
then the edge $u^{(\mathcal{W})}v^{(\mathcal{W})}$ cannot be inserted in the yellow region in $\mathcal{W}$, 
since it would cross $u^{(\mathcal{X})}v^{(\mathcal{X})}$ twice. 
Thus, $u^{(\mathcal{W})}v^{(\mathcal{W})}$ can only be inserted in the red region in $\mathcal{W}$.
If inserted in that region, $u^{(\mathcal{C})}v^{(\mathcal{C})}$ cannot be inserted in the yellow region in $\mathcal{C}$, 
since it would cross $u^{(\mathcal{W})}v^{(\mathcal{W})}$ twice.
In contrast, if the edge $u^{(\mathcal{X})}v^{(\mathcal{X})}$ is inserted in the \texttt{true} (green) region in $\mathcal{X}$, 
then $u^{(\mathcal{W})}v^{(\mathcal{W})}$ can be inserted in either of the two regions in $\mathcal{W}$. 
In particular, it can be inserted in the yellow region 
in a way such that $u^{(\mathcal{C})}v^{(\mathcal{C})}$ can then be inserted in any region in $\mathcal{C}$.

Finally, notice that if the target region in $\mathcal{C}$ is not the leftmost one, 
we can adapt the construction by leaving 
the region(s) to the left in $\mathcal{C}$ 
uncrossed by the wire gadget $\mathcal{W}$; see the clause gadget in the middle of Fig.~\ref{fig:red3SAT}.    
\end{proof}

Let $\phi$ be an instance of monotone \threesat and let $G(\phi)$ be the 
bipartite graph associated to $\phi$. 
Let $D(\phi)$ be a 2-page book drawing of $G(\phi)$ in which 

(i) all vertices lie on an horizontal line,
and from left to right, first the ones corresponding to negative clauses, then to 
variables, and finally to positive clauses; and 
(ii) the edges incident to vertices corresponding to positive clauses are drawn as circular arcs above that horizontal line, 
while the ones incident to vertices corresponding to negative clauses are drawn as circular arcs below it. 
In an slight abuse of notation, 
we refer to the vertices in $D(\phi)$ corresponding to variables and clauses simply as variables and clauses, respectively.

We construct a simple drawing $D'$ from $D(\phi)$ 
by first replacing the variables and clauses by  
variable gadgets and clause gadgets, respectively, and drawn in disjoint regions. 
Moreover, the clause gadgets corresponding to negative clauses are rotated $180^\circ$. 
We then insert the wire gadgets. 
The edges in $D(\phi)$ connecting variables to positive clauses 
are replaced by wire gadgets drawn as in the proof of Lemma~\ref{lem:wire}; see Fig.~\ref{fig:red3SAT}. 
Similarly, the edges in $D(\phi)$ connecting variables to negative  clauses 
are replaced by wire gadgets drawn as the ones before, but rotated $180^\circ$.  

We now describe how to draw the wire gadgets with respect to each other, 
so that the result is a simple drawing; see Fig.~\ref{fig:red3SAT} for a detailed illustration. 
First, we focus on the drawing locally around the variable gadgets. 
Consider a set of edges in $D(\phi)$ connecting a variable with some positive clauses. 
The drawing $D(\phi)$ defines a clockwise order of these edges around the common vertex starting from the horizontal line. 
We insert the 
corresponding wire gadgets locally around the variable gadget following this order. 
Each new gadget is inserted shifted up and to the right with respect to the previous one (as the blue and green gadgets depicted in Fig. \ref{fig:red3SAT}). 
Edges in $D(\phi)$ connecting a variable with some negative clauses 
are replaced by wire gadgets in an analogous manner with a $180^\circ$ rotation. 
We assign the three different regions in a clause gadget 
to the target regions in the wire gadgets following the 
rotation of the edges around the clause in $D(\phi)$. 
(Not that we can assume without loss of generality, by possibly duplicating variables, 
that each clause in $\phi$ contains three literals.)
Thus, locally around a clause gadget, 
it is then possible to draw the different wire gadgets connecting to it without crossing. 
Since $D(\phi)$ is a 2-page book drawing, 
the constructed drawing $D'$ is a simple drawing. 

Let $M$ be the set of $uv$ edges of all the gadgets. 
The fact that $\phi$ is satisfiable if and only if $M$ can be inserted into $D'$ follows now from Lemma~\ref{lem:wire}, finishing the proof of Theorem~\ref{thm:addk}.

\section{Maximizing the number of edges inserted is \APX-hard}
\label{sec:APXhard}

In this section we show that the maximization version of the problem of inserting missing edges from a prescribed set into a simple drawing is \APX-hard. 
This implies that, if $\P\neq \NP$,  then no \PTAS\  exists  for this problem. 
We start by showing that this maximization problem is \NP-hard.

\begin{theorem}
	Given a simple drawing $D(G)$ of a graph $G=(V,E)$ 
	and a set $M$
	of edges in the complement $\overline{G}$, 
	it is \NP-hard to find a maximum subset of edges $M'\subseteq M$ that extends  $D(G)$.
	\label{thm:max}
\end{theorem}

Our proof of Theorem~\ref{thm:max} is based on a reduction from the maximum independent 
set problem (MIS).
By showing that the reduction when the input graph has vertex degree at most three is actually a \PTAS-reduction 
we will then conclude that the problem is \APX-hard. 

An independent set of a graph $G=(V,E)$ is a set of vertices $S\subseteq V$ 
such that no two vertices in $S$ are incident with the same edge.
The problem of determining the maximum independent set (MIS) of a given graph is \APX-hard 
even when the graph has vertex degree at most three~\cite{APX-MIS-3}.
We first describe the construction of a simple drawing $D'(G')$ 
from the graph $G$ of a given MIS instance. 
Then we argue that for a well-selected set of edges $M$ that are not present
in $D'(G')$, 
finding a maximum subset $M' \subseteq M$ that can be inserted into $D'(G')$
is equivalent to finding a maximum independent set of $G$. 

\subsection{Constructing a drawing from a given graph}

\begin{figure}[tb]
	\centering 
	\subfloat[Vertex gadget $\mathcal{V}$.\label{fig:d1}]{
		\includegraphics[page=1]{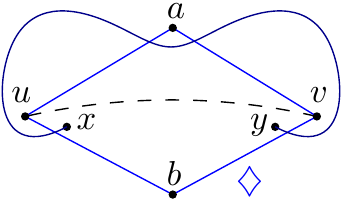}
	}
	\hspace{5em}
	\subfloat[Edge gadget $\mathcal{E}$.\label{fig:d2}]{
		\includegraphics[]{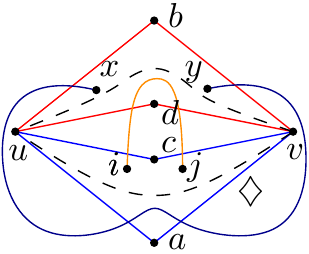}
	}
	
	\subfloat[\label{fig:P2} Two vertex gadgets interlinked by an edge gadget.]{
		\includegraphics[page =1]{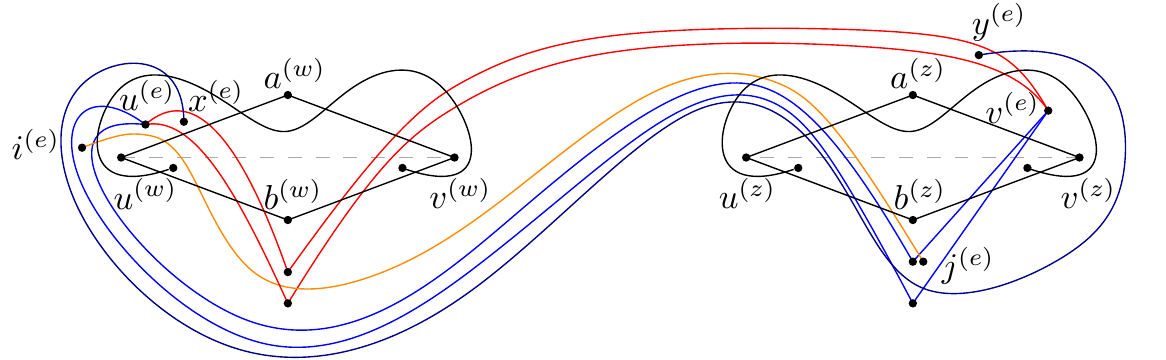}
	}
	\caption{Basic gadgets and drawings for the proof of Theorem~\ref{thm:max}.}\label{fig:basic_gadgets}	
\end{figure}
We begin by introducing our two basic gadgets, 
the vertex gadget $\mathcal{V}$ and the edge gadget $\mathcal{E}$, shown in Fig.~\ref{fig:basic_gadgets}. 
They are reminiscent of the gadgets in the previous section, 
but adapted to this different reduction. 
Similarly as in the previous gadgets, 
there is only one region in which the edge $uv$ can be inserted into $\mathcal{V}$ and 
only two regions in which the edge $uv$ can be inserted into $\mathcal{E}$. 
These regions are the ones in which the dashed arcs in Fig.~\ref{fig:d2} are drawn.

In Fig.~\ref{fig:P2} we combined an edge gadget and two vertex gadgets.  
This figure shows a copy $\mathcal{E}^{(e)}$ of the gadget $\mathcal{E}$ (that corresponds to an edge $e=wz$) 
drawn over two different copies, $\mathcal{V}^{(w)}$ and $\mathcal{V}^{(z)}$, of the gadget $\mathcal{V}$ (that  correspond to vertices $w$ and $z$, respectively). 
We relabel the vertices in the copies of these gadgets by using the vertex or edge to which they correspond as  their superscripts.
Since there is only one region in which $v^{(w)}u^{(w)}$ and $v^{(z)}u^{(z)}$ can be drawn, 
inserting both of these edges prevents $v^{(e)}u^{(e)}$ from being inserted.
Inserting either only $v^{(w)}u^{(w)}$ or only $v^{(z)}u^{(z)}$ leaves exactly one possible region where  $v^{(e)}u^{(e)}$ can be inserted.

\begin{sidewaysfigure}
	\centering
	\includegraphics[width=1\columnwidth,page=2]{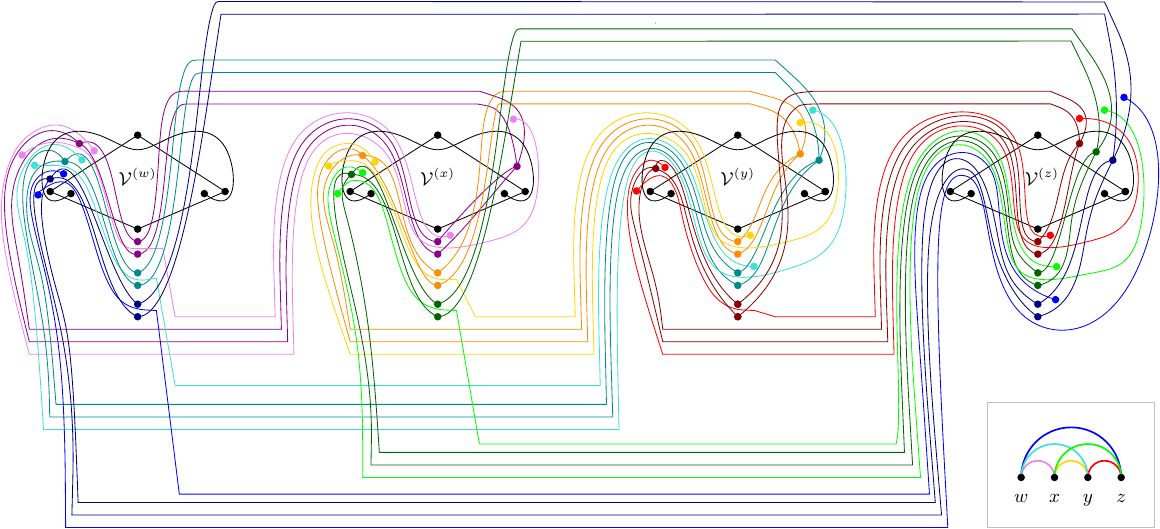}
	\caption{\label{fig:K4} Drawing obtained by a reduction from $K_4$.}
\end{sidewaysfigure}

We have 
all the  ingredients needed for our construction.
Suppose that we are given a simple graph $G = (V,E)$. 
This graph admits a 1-page book drawing $D(G)$ in which the vertices are placed on a horizontal line and the edges are drawn as circular arcs in the upper halfplane. 
Since the edge gadget does not interlink the vertex gadgets symmetrically, 
we consider the edges in $D(G)$ with an orientation from their left endpoint to their right one.

The following lemma shows that is possible to replace 
each vertex $w\in V$ in the drawing by a vertex gadget $\mathcal{V}^{(w)}$ 
and each edge $e\in E$ by an edge gadget $\mathcal{E}^{(e)}$,
and obtain simple drawing $D'(G')$ 
(where $G'$ is the disjoint union of the underlying graphs of the vertex- and  edge gadgets).

\begin{lemma}\label{lem:simultaneous_d2}
	Given a 1-page book drawing 
	$D(G)$ of a graph $G = (V,E)$, 
	then we can replace every vertex by a vertex gadget and every edge by an edge gadget to obtain a simple drawing.	
\end{lemma}
\begin{proof}
	We show that the copies $\{\mathcal{E}^{(e)}\,:\, e\in E\}$ can be inserted into $\bigcup_{w\in V} \mathcal{V}^{(w)}$ such that such that vertex gadgets corresponding to different vertices are drawn in disjoint regions and 
	for every edge $e=wz\in E$, $\mathcal{V}^{(w)}\cup \mathcal{V}^{(z)}\cup \mathcal{E}^{(e)}$ is as in Fig.~\ref{fig:P2} (up to interchanging the indices $w$ and $z$), and such that the resulting drawing is simple.
	
	First, for each vertex $w\in V$ we place the gadget $\mathcal{V}^{(w)}$ in its position, so all the 
	copies of $\mathcal{V}$ lie (equidistant) on a horizontal line and do not cross each other.
	For the edges of $G$, since the drawing in Fig.~\ref{fig:P2} is not symmetric, 
	we choose an orientation.
We orient all the edges in the 1-page book drawing $D(G)$ of $G$ from left to right.	
	We start by inserting the corresponding $\mathcal{E}$ gadgets
	from left to right
	and from the shortest edges in $D(G)$ to the longest. 
	For an edge $wz$, the intersections of the gadget $\mathcal{E}^{(wz)}$: 
	(i) with the edges $u^{(w)}a^{(w)}$ and $u^{(w)}b^{(w)}$  are placed to the left of all the previous intersections of other edge gadgets with that edge; 
	(ii) with the edge $v^{(w)}b^{(w)}$ are placed to the right of all the previous intersections with that edge;
	(iii) with the edge $v^{(w)}a^{(w)}$ are placed to the right of previous intersections with gadgets $\mathcal{E}^{(wt)}$
	and to the left of  previous intersections with gadgets $\mathcal{E}^{(tw)}$;
	(iv)  with the edges $u^{(z)}a^{(z)}$ and $u^{(z)}b^{(z)}$ are placed to the left of the previous intersections with gadgets $\mathcal{E}^{(tz)}$; 
	(v) with the edge $v^{(z)}b^{(z)}$ are placed to the left of all previous intersections; and
	(vi) with the edge $v^{(z)}a^{(z)}$ are placed to the left of all previous intersections with gadgets $\mathcal{E}^{(tz)}$; 
	see Fig.~\ref{fig:K4}.
	
	Moreover, the arcs of an edge gadget connecting two vertex gadgets are drawn 
	either completely in the upper half-plane or completely in the lower one with respect to the horizontal line 
	and two arcs cross at most twice.  
	If they are part of edges in edge gadgets connected to the same vertex gadget, 
	they might cross locally around this vertex gadget. 
	However, after this crossing, they follow the circular-arc routing induced by $D(G)$ (or its mirror image) and do not cross again. 
	Otherwise, with respect to each other, they follow the circular-arc routing induced by $D(G)$ (or its mirror image) 
	and thus cross at most once; see Fig.~\ref{fig:K4}.
	
	Since in neither of the gadgets two incident edges cross, 
	and edges of different gadgets are vertex-disjoint, 
	we only have to worry about edges from different gadgets crossing more than once.
	By construction, no edge in an edge gadget intersects more than once with an edge in a vertex gadget.
	Thus, it remains to show that any two edges 
	from two distinct edge gadgets cross at most once. 
	Such two edges are included in a subgraph $H$ of $G$ with exactly four vertices. The drawing induced by the four vertex gadgets and the at most six edge gadgets is homeomorphic to a subdrawing of the drawing in Fig.~\ref{fig:K4}. 
	It is routine to check that it is a simple drawing, and thus any two edges  cross at most once. 
\end{proof}

\subsection{Reduction from maximum independent set}
\label{subsection:reduction_from_max_inde}

\begin{proof}[of Theorem~\ref{thm:max}] 
Given a graph $G=(V,E)$,  
we reduce the problem of deciding whether $G$ has an independent set of size $k$
to the problem of deciding whether 
the simple drawing $D'(G')$ constructed as in Lemma~\ref{lem:simultaneous_d2}
with a candidate set of edges $M$ 
(where $M = \{u^{(w)} v^{(w)}: w\in V\} \cup \{u^{(e)} v^{(e)}: e \in E\}$) 
can be extended with a set of edges $M'\subseteq M$ of  
cardinality $|M'|= |E| + k$.

	To show the correctness of the (polynomial) reduction, 
	we first show that if 
	$G$ has an independent set $I$ of size $k$, then we can extend $D'(G')$ with a set $M'$ of $|E| + k$ edges of $M$. 
	Clearly, the $k$ edges $\{u^{(w)}v^{(w)} : w \in I\}$ can be inserted into $D'(G')$ by the construction of the drawing. 
	Since $I$ is an independent set, each edge has at most one endpoint in $I$.
	Thus, in every edge gadget $\mathcal{E}^{(e)}$ at most one of the two possibilities 
	for inserting the edge $u^{(e)}v^{(e)}$ is blocked by the previous $k$ inserted edges.
	We therefore can also insert the $|E|$ edges $\{u^{(e)}v^{(e)}: e \in E\}$.
	
	Conversely, let $M'\subset M$ be a set of $|E| + k$ edges
	can be inserted into $D'(G')$ and that contains the minimum number of $uv$ edges from vertex gadgets.
	If the set of vertices $\{w \in V: u^{(w)}v^{(w)} \in M'\}$
	is an independent set of $G$, then we are done, 
	since at most $|E|$ edges of $M'$ can be from edge gadgets,
	so at least $k$ are from vertex gadgets.
	Otherwise, there are two edges $u^{(w)}v^{(w)}$ and $u^{(z)}v^{(z)}$ in $M'$ 
	such that the corresponding vertices $w,z\in V$ are connected by the edge
	$wz \in E$. 
	By the construction of $D'(G')$ this implies that the edge $u^{(wz)}v^{(wz)}$
	belongs to $M$, but it cannot be in $M'$.
	By removing the edge $u^{(w)}v^{(w)}$ and inserting the edge $u^{(wz)}v^{(wz)}$ into $D'(G')$,
	we obtain another valid extension with the same cardinality
	but one less $uv$ edge from a vertex gadget. 
	This contradicts our assumption. 
\end{proof}

The presented reduction can be further analyzed to show that the problem is actually \APX-hard. 
Note that the problem we are reducing from, 
maximum independent set in simple graphs, is \APX-hard~\cite{APX-MIS-3} even in graphs with vertex degree at most three.  
Our reduction can be shown to be an \L-reduction inthat case, implying a \PTAS-reduction. 
This shows the following result (details are provided in Appendix~\ref{ap:apx}): 

\begin{corollary}
	\label{cor:apx-hard}
	Given a simple drawing $D(G)$ of a graph $G$ and a set of edges $M$ of the complement of $G$, 
	finding the size of the largest subset of edges from $M$ extending $D(G)$ is \APX-hard.
\end{corollary}

\section{Inserting one edge in a simple drawing}
\label{sec:adding_one_edge}

In this section, we consider the problem of extending a simple drawing of a graph by inserting exactly one edge $uv$ for a given pair of  non-adjacent vertices $u$ and $v$.
We start by rephrasing our problem as a problem of finding a certain path in
 the dual of the planarization of the drawing. 

Given a simple drawing $D(G)$ of a graph $G=(V,E)$, 
the {\em dual graph} $G^*(D)$ has a vertex corresponding to each cell of $D(G)$ 
(where a cell is a component of $\mathbb{R}^2\setminus D(G)$). 
There is an edge between two vertices if and only if 
the corresponding cells are separated by the same segment of an edge in $D(G)$.  
Notice that $G^*(D)$ can also be defined as the plane dual of the planarization  of $D(G)$, 
where crossings are replaced by vertices so that the resulting drawing is plane. 

We define a coloring $\chi$ of the edges of $G^*(D)$ by labeling the edges of the original graph $G$ using numbers from $1$ to $|E|$, and assigning 
 to each edge of $G^*(D)$ 
the label of the edge  that separates the cells corresponding to its incident vertices. 
Given two vertices $u,v\in V$, 
let $G^{*}(D,\{u,v\})$ be the subgraph of $G^*(D)$ obtained by removing the edges 
corresponding to connections between cells separated by an (arc of an) edge incident to $u$ or to $v$, 
and let $\chi'$ be the coloring of the edges coinciding with $\chi$ in every edge. 
The problem of extending $D(G)$ with one edge $uv$ is equivalent to
the existence of a heterochromatic path in $G^{*}(D,\{u,v\})$ (i.e., no color is repeated) with respect to $\chi$, 
between two vertices that corresponds to a cell incident to $u$ and a cell incident to $v$, respectively.

We remark that, from this dual perspective, 
it is clear that the problem of deciding if a simple drawing can be extended with a given set of edges is in \NP. 

The general problem of finding an heterochromatic path in an edge-colored graph is \NP-complete,  
even when each color is assigned to at most two edges. 
The proof can be found in Appendix~\ref{sec:hetero}.

\begin{theorem}
\label{thm:np_heterochromatic}
Given a (multi)graph $G$ with an edge-coloring $\chi$ 
and two vertices $x$ and $y$, 
it is NP-complete to decide whether there is a heterochromatic path in $G$ from $x$ to $y$,  
even when each color is assigned to at most two edges.
\end{theorem}

However, in  our setting the multigraph and the coloring come from a simple drawing.  
The following theorem  
shows a particular case in which we can decide in polynomial time if an edge can be inserted. 

\begin{theorem}
\label{thm:alg_one_extension}
Let $D(G)$ be a simple drawing of a graph $G=(V,E)$ and let $u$, $v\in V(G)$  be non-adjacent vertices. 
If $\{u,v\}$ is a dominating set for $G$, that is, every vertex in $V\setminus\{u,v\}$ is a neighbor of  $u$ or $v$, then the problem of extending $D(G)$ with the edge $uv$ can be decided in polynomial time.
\end{theorem}

An algorithm proving this result can be found in Appendix~\ref{ap:one_edge}. 
We sketch here the idea. 
\begin{figure}[tb]
	\centering
	\includegraphics[page=1]{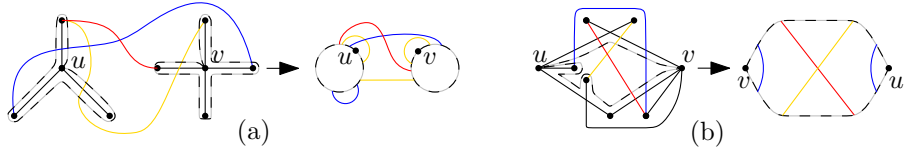}
	\caption{Reduction to the path problem with holes.}
	\label{fig:pph}
\end{figure}
The first step is to reduce our problem 
to the {\em path problem with holes} (PPH): 
Given two open disks $h_1, h_2\subseteq \mathbb{R}^2$ whose closures (called holes) are either disjoint or they coincide $h_1=h_2$, 
a set $\mathcal{J}$ of colored Jordan curves in $\Gamma = \mathbb{R}^2\setminus (h_1\cup h_2)$, and two distinct points $p$, $q\in \Gamma\setminus \bigcup \mathcal{J}$, 
we want to decide if there is a $pq$-arc intersecting at most one arc in $\mathcal{J}$ from each color. 
If $h_1=h_2$, we say that the instance of the PPH has one hole. 

Consider the subdrawing $D_{u,v}$ of $D(G)$ 
consisting of $u$, $v$, all vertices adjacent to them and all the edges incident to $u$ or to $v$. 
Fig.~\ref{fig:pph} illustrates the reduction from the problem of inserting $uv$ in $D_{u,v}$ to the PPH. 
Based on our reduction, one can make further assumptions on any instance $(\Gamma, \mathcal{J}, p, q)$ that we consider of the PPH problem:
(i) for every two different arcs $\alpha_1$, $\alpha_2\in \mathcal{J}$, $|\alpha_1\cap \alpha_2|\leq 1$;
(ii) pairs of arcs in $\mathcal{J}$ with the same color do not cross; 
and (iii) each arc in $\mathcal{J}$ has both ends on the union of the boundaries of the holes $\partial h_1\cup \partial h_2$.

Given an instance $(\Gamma,\mathcal{J},p,q)$ of the PPH, 
an arc $\alpha\in J$ is {\em separating} if $p$ and $q$ are on different connected components of $\Gamma \setminus \alpha$. 
We divide the arcs in  $\mathcal{J}$ into three different types: 
(T1) arcs with ends on different holes; 
(T2)  separating arcs with ends on the same hole; 
and (T3) non-separating arcs with ends on the same hole. 

Arcs of type T3 can be preprocessed with the operation that we denote \emph{enlarging one hole using $\alpha$}, as showed in Fig.~\ref{fig:two_op} (a). 
\begin{figure}[tb]
	\centering
	\includegraphics[]{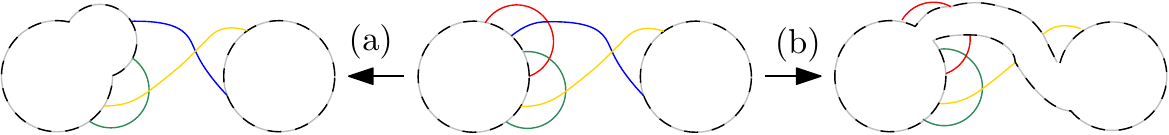}
	\caption{Transforming an instance of PPH: (a) enlarging a hole along an arc and (b) cutting through an arc.}
	\label{fig:two_op}
\end{figure} 
Once all the arcs in $\mathcal{J}$ are of type either T1 or T2,  
the algorithm determines the existence of a feasible $pq$-arc based on the colors of the arcs in $\mathcal{J}$. 
If all the arcs have different colors we have a solution. 
Otherwise we consider two arcs of the same color. 
If both arcs are of type T2, then there is no valid $pq$-arc and our algorithm stops.  
For handling the cases in which at least one of these arcs is of type T1, 
the idea is to try to find a solution that does not cross it. 
To do so, we use the operation denoted \emph{cutting through an arc} 
illustrated in Fig.\ref{fig:two_op} (b). 
If of the two arcs of the same color is of type T1 and the other is of type T2, 
there is a valid $pq$-arc if and only if there is a valid $pq$-arc after cutting through the T1 arc. 
Otherwise, if both are of type T1, 
there is a solution if and only if either there is a solution after cutting through the first arc 
or there is a solution after cutting through the second one. 
Note that the operation of cutting through an arc produces  
an instance with only one from an instance with two holes. 
This guarantees that the algorithm runs in polynomial time. 

\section{Conclusions}
In this paper we showed that given a simple drawing $D(G)$ of a graph $G=(V,E)$ 
and a prescribed set $M$ of edges of the complement of $G$, 
it is \NP-complete to decide whether $M$ can be inserted into $D(G)$. 
Moreover, it is \APX-hard to find the maximum subset of edges in $M$ that can be inserted into $D(G)$.  
We remark that the reduction showing \APX-hardness cannot replace the one 
showing \NP-hardness of inserting the whole set $M$ of edges, 
since, by construction, in the \APX-hardness reduction some of the edges in $M$ cannot be inserted.  

Focusing on the case $|M|=1$, 
we showed that a generalization of this problem is \NP-complete 
and 
we found sufficient conditions guaranteeing a polynomial-time decision. 
We hope that this paves the way to solve the following question.

\setcounter{problem}{0}
\begin{problem}
		Given a simple drawing $D(G)$ of a graph $G$ and a pair $u$, $v$ of non-adjacent edges, what is the computational complexity of deciding whether we can insert $uv$ into $D(G)$ 
		such that the result is a simple drawing? 
\end{problem}

\paragraph{Acknowledgments}

We want to thank the anonymous reviewers for their insightful comments. 

\bibliographystyle{splncs04}
\bibliography{refs_extending}

\appendix

\newpage

\section{Proof of Corollary~\ref{cor:apx-hard}}
\label{ap:apx}

\begin{proof} 
Since the MIS problem for graphs
with vertex degree at most three is \APX-hard~\cite{APX-MIS-3}, 
it suffices to show that the reduction proving Theorem~\ref{thm:max} is an  \L-reduction. 
This type of reductions was introduced by Papadimitriou and Yannakakis~\cite{papa91}. 
In order to provide a formal definition, we present some notation.

Given an \NP-optimization problem $P$, 
we denote by $I(P)$ the set of instances of $P$. 
For example, the set of all graphs is $I(\text{MIS})$. 
The \NP-optimization problem $P$ has associated an objective function $cost_P$ that we would like to either maximize or minimize (in our case maximize). 
For each instance $x\in I(P)$ 
we denote by $\text{opt}_P(x)$ the optimal value of a feasible solution with respect to $cost_P$. 
(For the MIS problem, the feasible solutions are the independent sets of the instance graph and $cost$ measures the size of a set.)

	Let $A$ and $B$ be a pair of \NP-optimization problems. 
	There is an {\em $L$-reduction} from $A$ to $B$ if there are polynomial-time computable functions $f$ and $g$ and positive constants $c_1$ and $c_2$ such that,
	\begin{enumerate}[(i)]
		\item $f$ maps every instance $x\in I(A)$ to an instance $x'=f(x)\in I(B)$;
		\item $g$ maps every feasible solution $y'$ of $x'=f(x)$ to a feasible solution $y=g(x,y')$ of $x\in I(A)$; 
		\item for every instance $x\in I(A)$, $\text{opt}_B(f(x))\leq c_1\cdot\text{opt}_A(x)$; and 
		\item for every instance $x\in I(A)$ and for every feasible solution $y'$ of $x'=f(x)$,  
		$|\text{opt}_A(x)-cost_A(y)|\leq c_2\cdot|\text{opt}_B(x')-cost_B(y')|$, 
		where $y=g(x,y')$.
	\end{enumerate}
	Given a simple graph $G=(V,E)$, 
	we construct a simple drawing $D'(G')$ as in Lemma \ref{lem:simultaneous_d2}.  This construction plays the role of $f$ in (i). 
	We denote by $M$ the candidate set of edges consisting of all the $uv$ edges of the gadgets used to construct $D'(G')$, that is, $M = \{u^{(w)} v^{(w)}: w\in V\} \cup \{u^{(e)} v^{(e)}: e \in E\}$. 
	Then, as argued in the proof of Theorem~\ref{thm:max}, 
	$G$ has an independent set of size $k$ if and only if 
	we can insert $|E| + k$ from $M$ into $D'(G')$. 
	Moreover, suppose that $M'\subseteq M$ is a subset of $|E|+k$ edges that can be inserted into $D'(G')$. 
	Using the ideas of the proof of Theorem~\ref{thm:max}, 
	if the set of vertices $\{w \in V: u^{(w)}v^{(w)} \in M'\}$
	is an independent set of $G$, then they are an independent set of $G$ of size $k$. 
	Otherwise, there are two edges $u^{(w)}v^{(w)}$ and $u^{(z)}v^{(z)}$ in $M'$ 
	and then the edge $u^{(wz)}v^{(wz)}$ cannot be in $M'$.
	By removing the edge $u^{(w)}v^{(w)}$ and inserting the edge $u^{(wz)}v^{(wz)}$ into $D'(G')$,
	we obtain another set of candidate edges that can be inserted with the same cardinality 
	but with one less $uv$ edge from a vertex gadget. 
	Iterating this process we obtain a subset of  $|E| + k$ edges $M''\subseteq M$ 
	such that the set of $k$ vertices $\{w \in V: u^{(w)}v^{(w)} \in M''\}$ is an independent set of $G$. 
	This defines the function $g$ mapping a feasible subset $M'\subseteq M$ of at least $|E|$ edges that we can insert into $D'(G')$ to an  
	independent set in $G$ of size $|M'|-|E|$. 
	We extend $g$, so that every feasible subset $M'\subseteq M$ with $|M'|\leq |E|$ is mapped to the empty set. 
	This proves (ii). 
.
	
	Let $\alpha=\alpha(G)$ be the size of the maximum independent set of $G$.  	
	We now show (iii). 
	First, observe that the handshaking lemma and the fact that the vertex degrees in $G$ are at most three imply $|E| \leq 3/2 |V|$. 
	We now bound $|V|$ in terms of $\alpha(G)$. 
	Wei~\cite{wei} and Caro~\cite{caro} independently showed that $\alpha(G)\geq \sum_{v\in V}1/(d(v) + 1)$, where $d(v)$ is the degree of vertex $v$. 
	Thus, in our case $|V|\leq 4\alpha$. 
	(This bound also follows from Tur\'an's theorem~\cite{turan1941external}.)
	Plugging this into the equation obtained by the handshaking lemma we get 
	$|E|\leq 3/2|V|\leq 6\alpha$. 
	Since an optimal solution for the problem of inserting the largest subset of candidate edges into $D'(G')$ has size $\alpha +|E|\leq 7\alpha$, we have proven (iii) for a constant $c_1 = 7$.

	Finally, we show (iv) for the constant $c_2=1$.
	Let  $M'\subseteq M$ be a set of $l$ edges that can be inserted into $D'(G')$. 
	If $l\leq |E|$, then $g$ maps $M'$ to the empty set and 
	we have that 
	$\alpha-0\leq |E|+\alpha-l$.  
	Otherwise, if $l=|E|+l'$ for $l'\geq 0$ we have that 
	$\alpha - l' = |E| + \alpha - |E| - l'$. 
	Thus, the absolute errors are in the worst case the same, as desired. 
\end{proof}

\section{Proof of Theorem~\ref{thm:np_heterochromatic}}
\label{sec:hetero}
\begin{proof}
	We reduce from \threesat. 
	Given a formula in 3-CNF with $n$ variables $x_1$, $\ldots$, $x_n$ and $m$ clauses $C_1,\ldots, C_m$,
	we construct an edge-colored (multi)graph $G$ as the one depicted in Fig.~\ref{fig:3sat}. 
	For each clause $C_j$, we construct a subgraph that consists of two vertices $sc_j$ and $tc_j$ joined by three different edges 
	with colors $j_1$, $j_2$, and $j_3$, respectively, 
	corresponding to the (without loss of generality) three literals in the clause.

	For each variable $x_i$, we construct a subgraph that consists of 
	two vertices $sx_i$ and $tx_i$, and two disjoint paths connecting them. 
	The first path has its initial edge colored with color $i$, 
	while the rest of the edges correspond to the literals $x_i$ in the clauses. 
	The second path also has its initial edge colored $i$, and the rest of the edges correspond to the literals $\neg x_i$ in the clauses. 
	If an edge corresponds to the $k$-th literal of the clause $C_j$ we assign color $j_k$ to this edge.

	We now join all the clause subgraphs by identifying $tc_j$ with $sc_{j+1}$, for $j=1,\ldots,m-1$. 
	We also join all the variable subgraphs by identifying $tx_i$ with $sx_{i+1}$, for $i=1,\ldots,n-1$. 
	Finally, we identify $tc_m$ with $sx_1$.
	
	It is easy to see that there is a heterochromatic path in $G$ from $sx_1$ to $tc_m$ if an only if the \threesat\ instance is satisfiable.
	Finally, notice that we can easily modify the reduction to construct a simple graph instead of a multigraph by  
	subdividing edges and using new colors.
\end{proof}

\begin{figure}[tb]
	\centering
	\includegraphics[page=2]{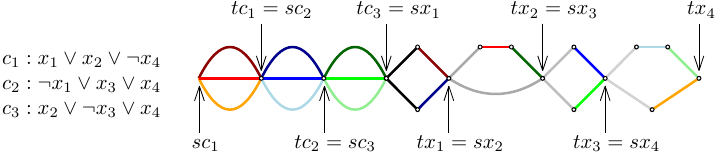}
	\caption{Reduction from \threesat: example with four variables and three clauses.}
	\label{fig:3sat}
\end{figure}

\section{Full proof of Theorem~\ref{thm:alg_one_extension}}
\label{ap:one_edge}

This section is devoted to prove Theorem~\ref{thm:alg_one_extension}. 
The first step is to reduce the problem of inserting the edge $uv$ to the problem of finding a valid path crossing some colored arcs at most once in a plane with some forbidden regions (holes). This new problem 
has  the advantage of being a more suitable ground for inductive proofs. 
The main ingredients needed in our algorithm are a series of lemmas 
describing sufficient conditions for which this problem has a solution. 

For an integer $k\geq 1$, a {\em plane with $k$ holes} 
is a set $\Gamma\subseteq \mathbb{R}^2$ obtained from 
considering $k$ disjoint simple closed curves in $\mathbb{R}^2$, all bounding a common cell, 
and removing for each curve $C$ the cell bounded by $C$ that is disjoint from the rest of the curves. If $k=1$, then only one side of $C$ is removed. 
The closure of each removed cell is a {\em hole of $\Gamma$}. 

\paragraph{Path problem with holes (PPH).} 
Given a plane with holes $\Gamma$ and a set of colored Jordan arcs $\mathcal{J}$ drawn in $\Gamma$, 
the path problem with holes asks whether there is a Jordan arc connecting two points $p$, $q\in \Gamma\setminus \mathcal{J}$, called \emph{terminals}, that 
crosses at most one arc in $\mathcal{J}$ of each color.
If such a $pq$-arc exists, then it is a {\em valid $pq$-arc} for the instance $(\Gamma, \mathcal{J},p,q)$. 

We assume that every instance of the path problem with holes that we  consider meets  the following properties:
\begin{enumerate}[(i)]
	\item Every two arcs of $\mathcal{J}$ share at most one point.
	\item Pairs of distinct arcs in $\mathcal{J}$  having the
	same color are disjoint. 
	\item Each arc in $\mathcal{J}$ starts and ends on the boundary of $\Gamma$, that is, no arc has an endpoint in the interior of $\Gamma$.
\end{enumerate}

\paragraph{Reduction.}
Let $D(G)$ be a drawing of  a graph $G$, and let $\{u,v\}$ be a dominating set of vertices in $G$ such that $uv$ is an edge of $\overline{G}$. 
We now reduce
the problem of deciding whether $uv$ can be inserted into $D(G)$
to the path problem with at most two holes. 

If $G'$ is a subgraph of $G$, then we denote by $D[G']$ the subdrawing of $D$ induced by the vertices and edges  of $G'$.
In a slight abuse of notation, if $G'$ consists only of a vertex $v$ or of an edge $e=uv$, then we will write $D[v]$ and $D[e]$ (or $D[uv]$), respectively.

For a vertex $v$ of $G$, the \emph{star of $v$} consists of $v$, its adjacent vertices, and its incident edges. 
Let $S_u$ and $S_v$ be the subdrawings of $D$ induced by the stars of $u$ and $v$, respectively. 
Moreover, let $H$ be the subgraph of $G$ that is the union of the stars of $u$ and $v$. 
Then, $S_u$ and $S_v$ are plane stars whose union is $D[H]$. 
If an extension with $uv$ exists, then the arc connecting $u$ and $v$ representing the edge $uv$ 
cannot cross any of those edges and must lie in the closure of a cell $F$ of $D[H]$ with $u$ and $v$ on its boundary.
Thus, our problem reduces to testing the existence of a valid $uv$-arc in each cell $F$ of $D[H]$ with both $u$ and $v$ on its boundary.

We can assume without loss of generality that 
$u$ and $v$ are incident to at least one edge 
by maybe inserting small segments incident to them. 
Let $F$ be a cell of $D[H]$ with both $u$ and $v$ on its boundary. 
Notice that it might be bounded or unbounded. 
Moreover, the part of $D[H]$ that is in the closure of $F$ 
can be connected of disconnected.  
If it is connected, we consider a simple closed curve $C$ in the interior of $F$, closely following 
the part of $D[H]$ that is in the closure of $F$. 
We slightly modify $C$ so that, at a certain occurrence of $u$ and of $v$ on $\partial F$, 
the curve $C$ touches $\partial F$; see the dashed curve in Fig.~\ref{fig:pph} (b).
In our reduction we consider all possible modifications of $C$, differing on where we decide to make $C$ touch $u$ and $v$. 
The number of possible resulting curves is at most the degree of $u$ times the degree of $v$. 
We define $C'=C$. 

If the part of $D[H]$ that is in the closure of $F$ is not connected 
it must consist of two connected components containing $u$ and $v$, respectively. 
We consider two simple curves $C$ and $C'$ in the interior of $F$ each one closely following 
one of these connected components. 
As before, we slightly modify the curves so that, at a certain occurrence of $u$ and of $v$ on $\partial F$, 
they touch $\partial F$; see the dashed curves in Fig.~\ref{fig:pph} (a).

In both cases, we consider the inside of the curves $C$ and $C'$ 
to be the regions bounded by them and such that the union of their closures contains $S_u\cup S_v$. 
Let $\Gamma$ be the closure of the region consisting of $F$ with the inside of the curves $C$ and $C'$ removed. 
Then, $\Gamma$ is a plane with at most two holes 
(the closures of the inside of the curves $C$ and $C'$). 

To finish our reduction, we need to identify the set of colored Jordan arcs and the two terminals in the path problem with holes. 
The set $\mathcal{J}$ is the defined as the union of the arcs of $D[e]\cap \Gamma$, for each edge $e \in E$. 
In order to assign colors to the arcs in $\mathcal{J}$, 
we first assign a different color to each edge of $G$. 
Each arc of $D[e]\cap \Gamma$ then inherits the color of $e$;  see Fig.~\ref{fig:pph}. 
Finally, the terminals $p$ and $q$ are points in the two cells of $C\cup \mathcal{J}$ in $\Gamma$ having $D[u]$ and $D[v]$ on their boundary, respectively.

\setcounter{figure}{5} 
\begin{figure}[tb]
	\centering
	\includegraphics[page = 1]{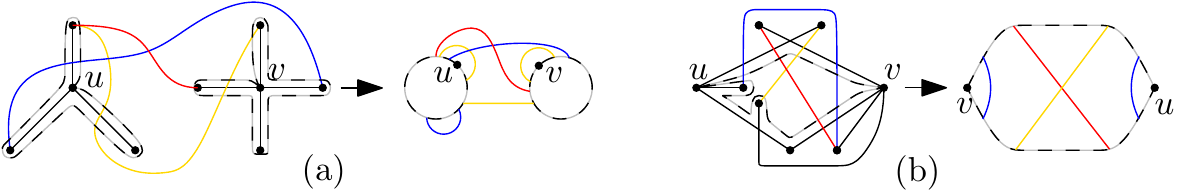}
	\caption{Reduction to the path problem with holes.}
	\label{fig:pphA}
\end{figure}

Notice that a reduction from  
the problem of inserting an edge into a simple drawing 
to the path problem with holes results in an instance satisfying properties (i) and (ii). 
Moreover, if $\{u,v\}$ is a dominating set for $G$, 
then the instance of the path problem with holes also meets property (iii).
The discussion above leads to the following statement:

\begin{observation}\label{obs:reduction}
	Let $D(G)$ be a simple drawing of a graph $G=(V,E)$ 
	and let $u$, $v\in V$  be non-adjacent vertices 
	such that $\{u,v\}$ is a dominating set for $G$. 
	The problem of deciding whether $uv$ can be inserted into $D(G)$ 
	can be reduced to the path problem with at most two holes. 
\end{observation}

We now prepare the tools for solving in polynomial time an instance of the path problem with at most two holes with properties (i)--(iii). 
Apart from introducing the notation and operations used in the algorithm solving that problem, 
we will show that if all arcs are of different colors, then there is always a solution.

Given a plane with holes $\Gamma$ and a set of Jordan arcs $\mathcal{J}$ in $\Gamma$, 
a {\em cell} of $(\Gamma, \mathcal{J})$ is the interior of a component of 
$\Gamma\setminus \mathcal{J}$. 
For any arc $\alpha\in \mathcal{J}$, a {\em segment of $\alpha$} is the closure of a component of $\alpha\setminus ( \mathcal{J}\setminus \{\alpha\})$. 
If the set of arcs has one element, $\mathcal{J} = \{\alpha\}$, then,
we abuse notation by writing $(\Gamma, \alpha)$ instead of $(\Gamma, \{\alpha\})$. 
Two cells of $(\Gamma,\mathcal{J})$ are {\em adjacent} if they share a segment of an arc in $\mathcal{J}$. 
Given two points $p,q\in \Gamma$ and a Jordan arc $\alpha$, $\alpha$ is {\em $pq$-separating} if every $pq$-arc in $\Gamma$ intersects $\alpha$. 

In the following, let $(\Gamma,\mathcal{J},p,q)$ be an instance of the path problem with at most two holes and properties (i)--(iii). %
Then, a $pq$-separating arc $\alpha\in \mathcal{J}$ has its ends on the same hole of $\Gamma$ 
and $p$ and $q$ are in different cells of $(\Gamma, \alpha)$. 
Moreover, each arc $\alpha\in \mathcal{J}$ is one of the  following three types:

\begin{enumerate}
	\item[\textbf{T1}:] $\alpha$ has its ends on two different holes of $\Gamma$;
	\item[\textbf{T2}:] $\alpha$ has its ends on the same hole of $\Gamma$ and is $pq$-separating; and
	\item[\textbf{T3}:] $\alpha$ has its ends on the same hole of $\Gamma$ and is not $pq$-separating.
\end{enumerate}

We say that two instances of a problem are \emph{equivalent} if 
the lead to the same output of a decision problem. 
The following operation shows how to transform any instance $(\Gamma,\mathcal{J},p,q)$ into another equivalent one where no arcs of type T3 occur. 

\paragraph{Enlarging a hole along an arc.}
If there is an $\alpha\in \mathcal{J}$ such that $\alpha$ is of type T3, having both its ends on the same hole $h$, 
then the operation of {\em enlarging a hole along $\alpha$} converts $(\Gamma,\mathcal{J},p,q)$ into a new instance $(\Gamma',\mathcal{J}',p,q)$, 
where $\Gamma'$ is obtained from $\Gamma$ by removing the cell of $(\Gamma, \alpha)$ disjoint from $p$ and $q$ and 
$\mathcal{J}' = \mathcal{J} \cap \Gamma'$; 
see Fig.~\ref{fig:two_op} for an illustration.

\begin{figure}[tb]
	\centering
	\includegraphics[]{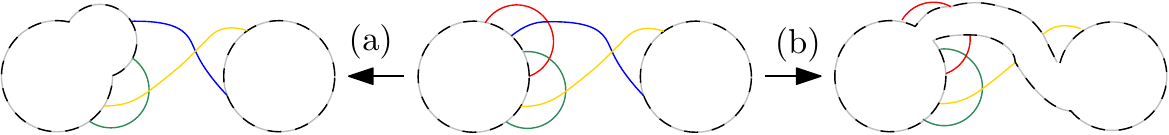}
	\caption{The two operations transforming an instance of the path problem with at most two holes: (a) enlarging a hole along an arc and (b) cutting through an arc. }
	\label{fig:two_opA}
\end{figure}

\begin{lemma}
	\label{lem:enlarging_hole}
	Let $(\Gamma,\mathcal{J},p,q)$ be an instance of the path problem with at most two holes and that meets properties (i)--(iii) 
	and let $(\Gamma',\mathcal{J}',p,q)$ be the instance obtained from $(\Gamma,\mathcal{J},p,q)$ by enlarging a hole along an arc $\alpha$ of type T3. 
	Then, for every arc $\beta \in \mathcal{J}$ there is at most one arc $\beta' \in \mathcal{J'}$ and it is of the same type as $\beta$. 
	Thus, $|\mathcal{J}'|<|\mathcal{J}|$.
	Moreover, there is a valid $pq$-arc in $(\Gamma,\mathcal{J},p,q)$ if and only if there is a valid $pq$-arc in $(\Gamma', \mathcal{J}',p,q)$. 
\end{lemma}

\begin{proof}
	To see that the first part holds, consider an arc $\beta\in \mathcal{J}\setminus \{\alpha\}$ with $\beta \cap \Gamma' \neq \emptyset$. 
	Let $F$ be the cell of $(\Gamma, \alpha)$ disjoint from $p$ and $q$. 
	If $\alpha\cap \beta=\emptyset$, then $\beta\cap \Gamma'=\beta'$. 
	Thus, the remaining case is that $\alpha$ and $\beta$ cross, 
	and because they can only cross once, 
	$\beta\setminus \alpha$ has two components: 
	one is included in $F$, while the other is in $\Gamma\setminus F=\Gamma'$ and its closure is $\beta'$. 
	Moreover,
	since $p$ and $q$ are not in $F$, 
	$p$ and $q$ belong to the same cell of $(\Gamma,\beta)$ if and only if they belong to the same cell of $(\Gamma',\beta')$. 
	Also $\beta$ has its ends on different holes if and only $\beta'$ has its ends on different holes of $\Gamma'$. 
	Now the second part of the lemma follows from the fact that any valid $pq$-arc in $(\Gamma,\mathcal{J})$ does not cross $\alpha$, since $p$ and $q$ are in the same cell of $(\Gamma, \alpha)$.
\end{proof}

The operation of enlarging a hole along an arc allows us to eliminate all arcs of type T3. 
Thus, if our instance has only one hole, then we can transform it to one where there are only arcs of type T2.
If there are two arcs of type T2 of the same color, then it is clear that there cannot be a solution.
The following result shows that this condition is also sufficient for instances with only one hole.

\begin{lemma}
	\label{lemma:one_hole}
	Let $(\Gamma,\mathcal{J},p,q)$ be an instance of the path problem with one hole 
	that meets properties (i)--(iii).   
	Then a valid $pq$-arc exists if and only if there are no two $pq$-separating arcs of the same color.
\end{lemma}

\begin{proof}
	Suppose that $\mathcal{J}$ has at most one $pq$-separating arc of each color.
	To show that there is a valid $pq$-arc, we proceed by induction on $|\mathcal{J}|$.  The base case $|\mathcal{J}|=0$ clearly holds. Henceforth, we assume  $|\mathcal{J}|\geq 1$. 
	
	If an arc in $\mathcal{J}$ is not $pq$-separating, then
	we apply Lemma \ref{lem:enlarging_hole} to reduce $(\Gamma,\mathcal{J},p,q)$ into an instance $(\Gamma',\mathcal{J}',p,q)$ with fewer arcs and satisfying the same conditions as $(\Gamma,\mathcal{J},p,q)$. 
	The induction hypothesis implies the existence of valid $pq$-arc in $(\Gamma',\mathcal{J}',p,q)$, that, by Lemma~\ref{lem:enlarging_hole}, also implies the existence of a valid one for $(\Gamma,\mathcal{J},p,q)$. 
	
	Suppose now that every arc $\mathcal{J}$ is $pq$-separating. 
	Since $|\mathcal{J}|\geq 1$, $p$ and $q$ are in different cells of $(\Gamma,\mathcal{J})$.  
	Let $F_p$ be the cell containing $p$ and let $\alpha\in\mathcal{J}$ be an arc with a segment $\sigma$ on the boundary of $F_p$. 
	Consider a point  $p'$ in the other cell of $(\Gamma,\mathcal{J})$ having $\sigma$ 
	on its boundary. 
	
	With the exception of $\alpha$, all the arcs in $\mathcal{J}$  are $p'q$-separating. 
	From the preceding discussion it follows that a valid $p'q$-arc not intersecting $\alpha$ exists.  
	No $p'q$-separating arc has the same color as $\alpha$ and therefore, 
	we can extend this valid $p'q$-arc to a valid $pq$-arc.
\end{proof}

With Lemma \ref{lemma:one_hole} in hand, we can now focus on instances with two holes. 
In this context, the condition of not having two $pq$-separating 
arcs of the same color is not sufficient to imply the existence of a valid $pq$-arc, as Fig.~\ref{fig:2holes} (a) shows. 
However, 
using the following operation, 
we can transform an instance with two holes into an instance with only one hole
when there is an arc of type T1 that cannot be crossed by a valid arc.

\paragraph{Cutting through an arc.} 
Let $\alpha\in\mathcal{J}$ be an arc of type T1 having its ends on distinct holes.  
The transformed instance $(\Gamma',\mathcal{J}',p,q)$ obtained from $(\Gamma,\mathcal{J},p,q)$ by {\em cutting through $\alpha$} is defined as follows.
Consider a thin open strip $\Sigma$ in $\Gamma$ covering $\alpha$ and neither containing $p$ nor $q$. 
Then, $\Gamma'=\Gamma\setminus \Sigma$ (this merges the two holes of $\Gamma$ into one hole) 
and $\mathcal{J}' = \mathcal{J}\cap \Gamma'$; 
see Fig.~\ref{fig:two_opA} (b) for an illustration.

\setcounter{figure}{7} 
\begin{figure}[tb]
	\centering
	\includegraphics[]{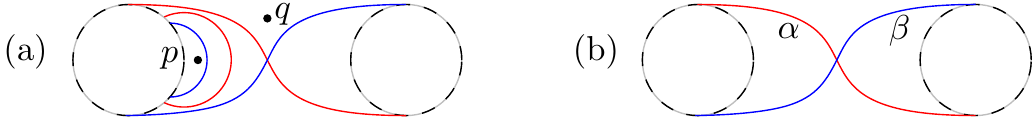}
	\caption{Relevant figures for instances with two holes.}
	\label{fig:2holes}
\end{figure}

\begin{observation}\label{obs:cutting}
	Let $(\Gamma,\mathcal{J},p,q)$ be an instance of the path problem with two holes  
	that meets properties (i)--(iii), 
	and let $(\Gamma',\mathcal{J}',p,q)$ be the instance obtained from $(\Gamma,\mathcal{J},p,q)$ by 
	cutting through an arc $\alpha\in \mathcal{J}$ of type T1. 
	Then, there is a valid $pq$-arc in $(\Gamma,\mathcal{J},p,q)$ 
	not crossing $\alpha$ if and only if there is a valid $pq$-arc in $(\Gamma',\mathcal{J}',p,q)$. 
\end{observation}

Suppose that $\alpha, \beta \in \mathcal{J}$ are two crossing arcs of type T1. 
Then,  
$(\Gamma,\{\alpha,\beta\})$
has exactly three cells; see Fig.~\ref{fig:2holes} (b). 
Moreover, if the terminals $p$ and $q$ are located in the pair of non-adjacent cells, 
then any valid $pq$-arc is forced to cross both $\alpha$ and $\beta$. 
The next result shows that if, for an arc $\alpha$ of type T1, 
there is no arc $\beta$ of type T1 producing this situation
and all arcs are of different colors, 
then there is a valid $pq$-arc not crossing $\alpha$.

\begin{lemma}
	\label{lemma:extend_avoid}
	Let $(\Gamma,\mathcal{J},p,q)$ be an instance of the path problem with two holes  
	that meets properties (i)--(iii). 
	Suppose that every arc in $\mathcal{J}$ is either of type T1 or of type T2 and 
	that all the arcs in $\mathcal{J}$ are of different colors.
	Let $\alpha\in\mathcal{J}$ be any arc of type T1.
	If, for every type T1 arc $\beta\in \mathcal{J}\setminus\{\alpha\}$ crossing $\alpha$,  
	$p$ and $q$ are in adjacent cells of $(\Gamma,\{\alpha, \beta\})$, 
	then there is a valid $pq$-arc not intersecting $\alpha$.  
\end{lemma}

\begin{proof}
	Let $(\Gamma',\mathcal{J}',p,q)$ be the instance obtained from cutting $\Gamma$ along $\alpha$.  
	Let $h$ be the hole of $\Gamma'$ obtained from merging the two holes $h_1$ and $h_2$ of $\Gamma$ with a thin strip covering $\alpha$. 
	We decompose the boundary of $h$ as the union of four arcs $\alpha_1$, $\gamma_1$, $\alpha_2$ and $\gamma_2$, 
	where $\alpha_1$ and $\alpha_2$ bound the strip covering $\alpha$ and, 
	for $i=1,2$, $\gamma_i$ is the arc on the boundary of $h_i$ connecting $\alpha_1$ and $\alpha_2$. 
	
	From Observation~\ref{obs:cutting}, it is enough to show the existence of a valid $pq$-arc in $(\Gamma',\mathcal{J}',p,q)$. 
	Assume for contradiction that there is no valid $pq$-arc for $(\Gamma',\mathcal{J}',p,q)$.  
	Lemma \ref{lemma:one_hole} shows that then there are two separating $pq$-arcs  $\beta_1$, $\beta_2\in\mathcal{J}'$ of the same color.
	Since all arcs in $\mathcal{J}$ are of different colors, 
	if $\beta\in \mathcal{J}\setminus \{\alpha\}$, then the arc components (one or two, depending on whether $\alpha$ and $\beta$ cross or not) of $\Gamma'\cap \beta$ induce one chromatic class of arcs in $\mathcal{J}'$. 
	Thus, there is an arc $\beta\in \mathcal{J}$ that crosses $\alpha$ and with two arc components $\beta_1$ and $\beta_2$ of $\Gamma'\cap \beta$.

	Since $\beta$ crosses $\alpha$, 
	each of $\beta_1$ and $\beta_2$ has exactly one endpoint on a different arc of $\alpha_1$ and $\alpha_2$. 
	By possibly relabeling $\beta_1$ and $\beta_2$, we may assume that, for $i=1,2$,  $\beta_i$ has an endpoint $a_i$ in $\alpha_i$.  
	For $i=1,2$, let $b_i$ be the endpoint of $\beta_i$ that is not $a_i$. 
	
	First, we suppose that both $b_1$ and $b_2$ are on the same hole of $\Gamma$, say  $h_1$ 
	(so $\beta$ is of type T2). 
	Then, $(\Gamma',\{\beta_1,\beta_2\})$ has three cells.
	As both $\beta_1$ and $\beta_2$ are $pq$-separating, 
	$p$ and $q$ are in the two cells of $(\Gamma',\{\beta_1,\beta_2\})$ that do not have $\gamma_2$ on the closure of their boundaries. 
	However, these two cells are included in the same cell of $(\Gamma, \beta)$, contradicting that $\beta$ is $pq$-separating (and thus of type T2). 
	
	Second, suppose that $b_1$ and $b_2$ are on different holes (so $\beta$ is of type T1).
	By symmetry, we may assume  $b_1\in h_1$ and $b_2\in h_2$. 
	There are three cells of $(\Gamma',\{\beta_1,\beta_2\})$, and, since $\beta_1$ and $\beta_2$ are $pq$-separating, 
	$p$ and $q$ are in the cells that have exactly one of $\beta_1$ and $\beta_2$ on their  boundary. 
	However, this implies that $p$ and $q$ are in non-adjacent cells of $(\Gamma,\{\alpha, \beta\})$, contradicting our hypothesis. 
\end{proof}

In fact, when all the arcs in $\mathcal{J}$ are of different colors there is always a valid $pq$-arc:

\begin{lemma}
	\label{lemma:colors}
	Let $(\Gamma,\mathcal{J},p,q)$ be an instance of the path problem with at most two holes that meets properties (i)--(iii).
	If all the arcs in $\mathcal{J}$ are of different colors, then there exists a valid $pq$-arc.
\end{lemma}

\begin{proof}
	If $\Gamma$ has only one hole, then the result follows from Lemma \ref{lemma:one_hole}, so we assume that $\Gamma$ has two holes. 
	We proceed by Induction on $|\mathcal{J}|$. 
	The base case $|\mathcal{J}|=1$ clearly holds. 
	Suppose $|\mathcal{J}|\geq 2$.
	
	If $\mathcal{J}$ has an arc of type T3, then
	we can apply Lemma~\ref{lem:enlarging_hole} to obtain an instance with fewer arcs that satisfies the same conditions as $(\Gamma,\mathcal{J},p,q)$. 
	The induction hypothesis shows that there is a valid $pq$-arc for the transformed instance, 
	and thus, there is also a valid $pq$-arc in $(\Gamma,\mathcal{J},p,q)$. 
	Henceforth, we assume that $\mathcal{J}$ has only arcs of types T1 and T2.
	
	Let $F_p$ be the cell of $(\Gamma,\mathcal{J})$ containing $p$ and let $\alpha\in \mathcal{J}$ be an arc having a segment $\sigma$ on the boundary of $F_p$. 
	Consider a point $p'$  in the cell adjacent to $F_p$ that has $\sigma$ on its boundary. 
	
	If $\alpha$ is of type T2 (with respect to terminals $p$ and $q$), then, as $\alpha$ is not $p'q$-separating, 
	applying Lemma~\ref{lem:enlarging_hole} as before shows that there is a valid $p'q$-arc (not crossing $\alpha$) 
	that can be extended to a valid $pq$-arc.
	
	Thus, the only remaining case is that $\alpha$  is of type T1, so we assume that $\alpha$ has its ends on two different holes $h_1$ and $h_2$.
	
	\begin{claim}
		Either there is a valid $pq$-arc not intersecting $\alpha$ or there is a valid $p'q$-arc not intersecting $\alpha$. 
	\end{claim}
	\begin{proof}
		Assume for contradiction that there are no valid $pq$- and $p'q$-arcs disjoint from $\alpha$. 
		Lemma \ref{lemma:extend_avoid} implies that there is a type T1 arc $\beta\in \mathcal{J}\setminus \{\alpha\}$ crossing $\alpha$,  
		such that the two non-adjacent cells $F_p^\beta$ and $F_q^\beta$ of $(\Gamma, \{\alpha,\beta\})$ contain $p$ and $q$, respectively. 
		Let $F^\beta$ be the other cell of $(\Gamma, \{\alpha,\beta\})$ neither including $p$ nor $q$. 
		Likewise, there exists $\beta'\in \mathcal{J}\setminus \{\alpha\}$ crossing $\alpha$, 
		such that the two non-adjacent cells $F_{p'}^{\beta'}$ and $F_q^{\beta'}$ of $(\Gamma, \{\alpha,\beta'\})$ contain $p'$ and $q$, respectively. 
		Let $F^{\beta'}$ be the other cell of $(\Gamma, \{\alpha,\beta'\})$.
		 
		Let $\times$ and $\times'$ be the crossings between $\alpha$ and $\beta$ and between $\alpha$ and $\beta'$, respectively.
		By symmetry, we may assume that when we traverse $\alpha$ from $h_1$ to $h_2$, we encounter $\times$ before $\times'$. 
		Also, by possibly relabeling $h_1$ and $h_2$, we may assume that $F_p^\beta$ has a subarc of the boundary of $h_1$ on its boundary, 
		while $F_q^\beta$ has a subarc of the boundary of $h_2$ on its boundary. 
		
		Since $p\in F_p$ and $F_p\subseteq F_p^\beta$, the segment $\sigma\subseteq \alpha$ shared by $F_p$ and $F_{p'}$ is located on $\alpha$ between the endpoint of $\alpha$ in $h_1$ and $\times$. 
		As $p'\in F_{p'}$,  both $p'$ and $F_{p'}$ are contained in  $F^\beta$. 
		
		The boundary of $F_{p'}^{\beta'}$ is a simple closed curve $C$  made of three arcs: 
		The first one connects $\times'$ to the boundary of $h_1$ along $\alpha$; 
		the second one 
		is an subarc of the boundary of $h_1$ connecting the endpoint of $\alpha$ on $h_1$ to the endpoint of $\beta'$ on $h_1$; 
		and the third one 
		is a subarc of $\beta'$ connecting the endpoint of $\beta'$ on $h_1$ to $\times'$. 
		Since $F_{p'}^{\beta'}$ contains $F_{p'}$, the points on $C\cap \beta'$ near $\times'$ are on the side of $\alpha$ that contains points 
		in $F_{p'}$. 
		As $\times'$ comes after $\times$ when we traverse $\alpha$ from $h_1$ to $h_2$, the points on $C\cap \beta'$ near $\times'$ are 
		in $F_q^\beta$. 
		Since the endpoint of $C\cap \beta'$ on $h_1$ is not in $F_q^\beta$, 
		the arc $C\cap \beta'$ crosses $\beta$ at some point $\times_{\beta,\beta'}$. 
		
		Since $\times_{\beta,\beta'}$ is the only crossing between $\beta$ and $\beta'$, the subarc of $\beta'$ from $\times'$ to $h_2$ is disjoint from $\beta$. 
		The points on this subarc near $\times'$ are 
		in $F^\beta$, and thus, the cell $F_q^{\beta'}$ is included in  $F^\beta$. 
		However, this shows that $F_q^\beta\cap F_q^{\beta'}=\emptyset$, contradicting that $q\in F_q^\beta\cap F_q^{\beta'}$. 
	\end{proof}
	
	From the previous claim, either there is a valid $pq$-arc not crossing $\alpha$ or there is a valid $p'q$-arc not crossing $\alpha$. 
	In the former case we are done. 
	In the later, we extend the valid $p'q$-arc to a valid $pq$-arc by crossing $\sigma$.
\end{proof}

With all the previous results 
we can now show the polynomial-time algorithm that proves Theorem \ref{thm:alg_one_extension}.
From Observation \ref{obs:reduction}, it is enough to solve the path problem with at most two holes for instances meeting properties (i)--(iii) 
in polynomial time. 
To show this we consider Algorithm~\ref{alg:dpp}. 

In Algorithm~\ref{alg:dpp}, $(\Gamma,\mathcal{J},p,q)$ is an instance of the path problem with at most two holes meeting properties (i)--(iii).  
$\text{ENLARGE}((\Gamma,\mathcal{J},p,q),\alpha)$ is a shorthand  for the instance obtained from $(\Gamma,\mathcal{J}, p, q)$ by enlarging a hole of $\Gamma$ along $\alpha$ and $\text{CUT}((\Gamma,\mathcal{J},p,q),\alpha)$ is a shorthand for the instance obtained from $(\Gamma,\mathcal{J}, p, q)$ by cutting through $\alpha$.   
We now show the correctness of Algorithm~\ref{alg:dpp}.

\begin{algorithm}[tb]
	\begin{algorithmic}[1]
		\caption{PPH$(\Gamma,\mathcal{J},p,q)$: outputs whether there is a valid $pq$-arc.}
		\label{alg:dpp}
		
		\WHILE{ $\mathcal{J}\neq \emptyset$} \label{stp:trivial}
		\IF{$\mathcal{J}$ has an arc $\alpha$ of type T3 (has its ends on same hole and is not $pq$-separating)} \label{stp:has_T3}
		\STATE $(\Gamma,\mathcal{J},p,q)\leftarrow \text{ENLARGE}((\Gamma,\mathcal{J},p,q),\alpha)$
		\ELSE \label{stp:no_T3}
		\IF{all arcs in $\mathcal{J}$ are of different colors} \label{stp:if_colors}
		\RETURN \texttt{True} \label{stp:colors_true}
		\ELSE
		\STATE find two arcs $\alpha$ and $\alpha'\in\mathcal{J}$ of the same color 
		\IF{both $\alpha$ and $\alpha'$ are of type T2 ($pq$-separating)} \label{stp:both_sep}
		\RETURN \texttt{False} \label{stp:both_sep_true}
		\ELSIF{$\alpha$ is of type T2 ($pq$-separating) and $\alpha'$ is of type T1 (has its ends on two holes)} \label{stp:mixed}
		\RETURN  $\text{PPH}( \text{CUT}((\Gamma,\mathcal{J},p,q),\alpha')))$ \label{stp:mixed_reduction}
		\ELSIF{both $\alpha$ and $\alpha'$ are of type T1 (have their ends on two holes)} \label{stp:both_T1}
		\RETURN $\text{PPH}(\text{CUT}((\Gamma,\mathcal{J},p,q),\alpha)) \vee  \text{PPH}(\text{CUT}((\Gamma,\mathcal{J},p,q),\alpha'))$ \label{stp:both_T1_reduction}
		\ENDIF
		\ENDIF
		\ENDIF
		\ENDWHILE
		\RETURN \texttt{True}
	\end{algorithmic}
\end{algorithm}

\begin{theorem}
	Let $(\Gamma,\mathcal{J},p,q)$ be an instance of the path problem with at most two holes that meets properties (i)--(iii). 
	Then, Algorithm~\ref{alg:dpp} decides whether there is a valid $pq$-arc 
	in polynomial time in the number of arcs in $\mathcal{J}$.
\end{theorem}

\begin{proof}
	Step \ref{stp:trivial}
	primarily checks if our current instance $(\Gamma,\mathcal{J},p,q)$ is trivial (i.e. $\mathcal{J}=\emptyset)$. 
	If not, the algorithm moves towards Step \ref{stp:has_T3}, where it verifies if $\mathcal{J}$ has a type T3 arc.  
	If it has one, it uses this arc to enlarge a hole and applies Lemma~\ref{lem:enlarging_hole} to update our instance to one with fewer arcs. 
	
	Otherwise, if $\mathcal{J}$ has no arcs of type T3, 
	the process continues with Step \ref{stp:no_T3}. 
	The first possibility is that all arcs in $\mathcal{J}$ are of different colors, and in this case the conditions of Lemma \ref{lemma:colors} apply, so there is a valid $pq$-arc (Steps \ref{stp:if_colors}--\ref{stp:colors_true}). 
	
	The second possibility is that $\mathcal{J}$ has two arcs $\alpha$ and $\alpha'$ of  the same color.
	If both $\alpha$ and $\alpha'$ are $pq$-separating, then clearly no valid $pq$-arc exists (Steps \ref{stp:both_sep}--\ref{stp:both_sep_true}). Otherwise,  one of them, say $\alpha'$, is of type T1. If $\alpha$ is $pq$-separating (type T2), then any valid $pq$-arc must cross $\alpha$, and thus, it does not cross $\alpha'$. 
	Therefore, it is enough to look for a valid $pq$-arc not crossing $\alpha'$. 
	Observation \ref{obs:cutting} translates that into finding a valid $pq$-arc for the instance with one hole that we obtain with the operation $\text{CUT}((\Gamma,\mathcal{J},p,q),\alpha')$ (Steps \ref{stp:mixed}--\ref{stp:mixed_reduction}).
	
	The third and last alternative is that both $\alpha$ and $\alpha'$ are of type T1. 
	In this case, any valid $pq$-arc crosses only one of $\alpha$ and $\alpha'$. 
	Thus, by Observation \ref{obs:cutting}, it is enough verify both instances obtained by applying the transformations $\text{CUT}((\Gamma,\mathcal{J},p,q),\alpha)$ and $\text{CUT}((\Gamma,\mathcal{J},p,q),\alpha')$ (Steps \ref{stp:both_T1}--\ref{stp:both_T1_reduction}). 
	An attentive reader may notice how, in principle,  an iterative occurrence of Step \ref{stp:both_T1_reduction} may
	lead into an exponential blow-up of the running time. 
	However, the fact that both instances that we obtain applying the transformations  $\text{CUT}((\Gamma,\mathcal{J},p,q),\alpha)$ and $\text{CUT}((\Gamma,\mathcal{J},p,q),\alpha')$ are instances of the path problem with one hole, guarantees that the algorithm goes through Step \ref{stp:both_T1_reduction} at most once.
\end{proof}

\end{document}